\documentclass[pdflatex,sn-mathphys-num]{sn-jnl}

\usepackage{graphicx}%
\usepackage{multirow}%
\usepackage[intlimits,sumlimits,namelimits]{amsmath}
\usepackage{amssymb,amsfonts}%
\usepackage{amsthm}%
\usepackage{mathrsfs}%
\usepackage[title]{appendix}%
\usepackage{xcolor}%
\usepackage{textcomp}%
\usepackage{manyfoot}%
\usepackage{booktabs}%
\usepackage{algorithm}%
\usepackage{algorithmicx}%
\usepackage{algpseudocode}%
\usepackage{listings}%
\usepackage[T1]{fontenc}
\usepackage{lmodern}
\usepackage{textcomp}
\usepackage{microtype}
\usepackage{xspace}
\usepackage{eucal}
\usepackage{eqlist} 
\usepackage[nosepfour,warning,np,debug,autolanguage]{numprint}
\usepackage{acro}
\usepackage{graphicx} \graphicspath{{./Figures/}}
\usepackage{subcaption}
\usepackage{geometry} 
\DeclareGraphicsExtensions{.pdf, .jpg, .tiff}
\usepackage{tikz}
\usetikzlibrary{trees}
\usepackage{pgfplots} \pgfplotsset{compat=1.18}
\usepackage{url}
\usepackage{titlesec}

\DeclareMathOperator*{\argmin}{\operatorname{argmin}}
\DeclareMathOperator*{\argmax}{\operatorname{argmax}}

\newcommand{\D}{\mathcal{D}}
\newcommand{\RR}{\mathbb{R}}
\newcommand{\oc}{\overline{c}}
\newcommand{\ox}{\overline{x}}
\newcommand{\ocs}{\overline{c}_s}
\newcommand{\osigma}{\overline{\sigma}}

\renewcommand{\eqref}[1]{(\ref{#1})}
\newcommand{\SISDG}[0]{SI$\cdot$DG}
\newcommand{\DSISDG}[0]{SI$\cdot$DDG}

\theoremstyle{thmstyleone}%
\newtheorem{theorem}{Theorem}
\newtheorem{lemma}[theorem]{Lemma}
\newtheorem{corollary}[theorem]{Corollary}

\theoremstyle{thmstyletwo}%

\theoremstyle{thmstylethree}%
\newtheorem{definition}{Definition}%
\newtheorem{problem}[definition]{Necessary Conditions}

\setcitestyle{authoryear}

\begin{document}

\title[Article title]{Social Distancing Equilibria in Games under Conventional SI Dynamics}
\author*[1]{\fnm{Connor D.} \sur{Olson}}\email{cjo5325@psu.edu}
\equalcont{These authors contributed equally to this work.}
\author[1,2]{\fnm{Timothy C.} \sur{Reluga}}\email{timothy@reluga.org}
\equalcont{These authors contributed equally to this work.}

\affil[1]{\orgdiv{Department of Mathematics}, \orgname{Pennsylvania State University}, \city{University Park}, \postcode{16802}, \state{PA}, \country{United States of America}}

\affil[2]{\orgdiv{Huck Institute of Life Sciences}, \orgname{Pennsylvania State University}, \city{University Park}, \postcode{16802}, \state{PA}, \country{United States of America}}

\abstract{
The mathematical characterization of social-distancing games in 
classical epidemic theory remains an important question, for their
applications to both infectious-disease theory and memetic theory.
We consider a special case of the dynamic finite-duration SI
social-distancing game where payoffs are accounted using
Markov decision theory without discounting, while distancing is constrained
by threshold-linear running-costs, and the running-cost of
perfect-distancing is finite. In this special case, we are able construct
strategic equilibria satisfying the Nash best-response condition explicitly by
integration. 
Our constructions are obtained using a new change of variables which simplifies
the geometry and analysis.
As it turns out, there are no singular solutions, and a time-dependent
bang-bang strategy consisting of a wait-and-see phase followed by a lock-down
phase is always the unique strategic equilibrium.
We also show that in a delay strategy space the bang-bang Nash equilibrium is an ESS, 
and that the optimal public policy exactly corresponds with the game equilibrium.
Social distancing only turns out to be a valuable tool in reducing outbreak
burden when the outbreak response delay falls in a Goldilocks zone, determined by the
speed of outbreak detection and the relative rolling cost of social distancing.
}


\maketitle

\section{Introduction}

The COVID-19 pandemic brought significant death and disruption to the lives of 
many people globally during 2020 and 2021. While the world waited for the
development of what was hoped to be a pandemic-ending vaccine,
social-distancing was a universally-accessible behavioral intervention
promoted by many health agencies. In principle, this was individual people
quarantining themselves on their own volition to avoid getting sick and
contributing to the epidemic. While health agencies desire their subjects to
follow their recommendations when applicable, and they do impose mandates
in exceptional circumstances like the COVID-19 pandemic,
the effectiveness of elective interventions like social distancing falls upon the 
willingness of people to maintain them. 

One approach to quantifying interventions like social distancing is to
analyze each individual's behavior in the context of the decisions of
others around them. To do this, we can appeal to classical game theory as a
model of decision problems. The social distancing game (SDG) is the most
prominent form of a continuous time (or differential) epidemic game. This type
of game represents contact-patterns as a strategy and considers the
consequences of each individual seeking to optimize their own well-being over
the course of an epidemic. One early manifestation of this game was developed by
\citet{bib:Reluga2010}, while a discrete time analogue was put-forward in
\citet{bib:Chen2012} in the parlance of a public avoidance game. 

Since this time, extensive work has been done on social distancing games, especially since 2020. 
Recent works in this vein include
\citep{bib:McAdams2020,bib:Farboodi2021,bib:Toxvaerd2020,bib:Choi2020,bib:Kordonis2022,bib:Carnehl2023,bib:McAdams2024,bib:McAdams2023,bib:Toxvaerd2024,bib:Avery2024,bib:Abel2024,bib:Schnyder2023,bib:Schnyder2025}.
These authors take many distinct approaches to studying social distancing,
from exploring the effects of different individual decision processes to modeling the impact of limited healthcare 
resources on how individuals behave. 
One key question which has not be well-addressed in the literature is whether the rational behaviors captured by these 
models are unique and under what conditions uniqueness fails. 
One serious attempt at addressing this question in the SIR SDG was made by \citet{bib:Reluga2013}, who focused on
the infinite-time horizon case but did not give a complete formal proof. 
One limitation of this work is the method of analysis is not easily generalizable to even a slightly 
more complicated game, so the general question of uniqueness in epidemic games is still unresolved. 
 
The recent paper by \citet{bib:Schnyder2025} also presents a uniqueness result for the SIR social
distancing game in the infinite-horizon case by demonstrating that the
Nash equilibrium behavior only depends on the current number of infected and the final number
of susceptibles, which also gives analytic expressions for the epidemic trajectory. However, this model deviates from 
Reluga's by accounting for the costs of infection and intervention using mathematically convenient 
but conceptually debatable convexity and information-bias. 
Accounting for costs in such a manner, potentially over modeling the 
costs in a realistic fashion, is a common theme in epidemic game theory and potentially impedes the development
of general theory which can further advance our understanding of the systems these games are meant to model. 
Just like we should strive to utilize epidemic models which accurately capture the disease dynamics in a real human population,
we should model costs in a fashion that reflect the way costs are incurred by individuals enduring an epidemic. 

The goal of this paper is to present the initial steps toward a complete
proof of the uniqueness of Nash equilibrium behavior in the SI variant of the
fixed-duration social distancing game.  Like \citet{bib:Reluga2013}, we work in
the framework developed by \citet{bib:Reluga2011}, with perfect information and
weak convexity hypotheses.  The scenario we imagine is an epidemic where
infection has a fixed cost, is easily diagnosed, permanent, and transmitted only through
peer-to-peer interactions.  The epidemic is
assumed to occur over a fixed window of time, at the end of which public health
efforts remove all further threat of infection, either through mass vaccination
or some other intervention.  While we intend the theory developed here to be
generally informative, we caution that it is NOT a perfect match to common
scenarios.  We revisit this in detail in the Discussion section, after the
model has been introduced and analyzed.
The SI model also has important applications in the spread of ideas and
information \citep{bib:Ross1911b,bib:Massad2013,bib:Dawkins1976}, and our
results are important for that field as well.

This paper proceeds in the following fashion. Our analysis begins with the
construction the SI social-distancing game.  Then we analyze this game under
the assumption of linear-threshold social-distancing costs for two different
strategy spaces.  First, using two assumptions about what the optimal behavior
should be, we reduce the strategy space to two-phase off-on strategies.  We
leverage this simplicity to establish the basic properties of the two-phase
off-on Nash equilibrium and prove it is an evolutionarily stable strategy (ESS)
\citep{bib:MaynardSmith1974} in the delay strategy space.  We are also able to
demonstrate the surprising result that the Nash behavior corresponds with the
socially optimal behavior, which reflects the discrete time result from
\citet{bib:Chen2012}.  Then for the second part, we introduce the
maximum-principle analysis, and use it to establish that two-phase off-on
strategies are the unique strategic equilibria of the SI social-distancing game
over the full time-dependent strategy space of a differential population game.

\section{SI Distancing Game}

We study the simplest epidemic with SI dynamics. Consider a population of
fixed size $N$ where people may become infected by a disease, and once
infected, they stay infected forever.  Every individual is hypothesized to be
either susceptible or infected -- there is no immunity or recovery. Let $t$ be
the date, $S(t)$ be the measure of the susceptible portion at $t$, while $I(t)$
is the measure of the infected population that is infected at $t$.  Let infection
occur by direct contact according to the law of mass-action with rate $\osigma \beta$,
where $\beta$ is the baseline transmission rate and $\osigma$ is the
relative overall reduction in transmission due to social-distancing ($\osigma \in [0,1]$).
Then the rates of change in the susceptible and infected are given as
\begin{subequations}
\label{eq:PopEpi}
\begin{align}
	\dot{S} &= - \osigma \beta SI,
	\\
	\dot{I} &= \osigma \beta SI.
\end{align}
Our epidemic starts at time $t=0$ with initial condition $S(0) = N - I_0, \, I(0) = I_0$.
\end{subequations}

Epidemic games hypothesize that there are trade-offs between the costs of
infection and the costs of prevention, and that individuals will choose
behaviors that balance these trade-offs.  Suppose the cost of getting infected
at date $t$ is $C_i(t)$, while a susceptible individual spends at rate $c_s(t)$
to prevent infection.  We define a relative infection risk for an individual
$\sigma:\mathbb{R}^+\rightarrow[0,1]$ as a function of their spending rate,
such that $\sigma(0)=1$ and $\sigma$ is decreasing.  When the population is
spending at average per capita rate $\ocs$, then the relative overall transmission
reduction $\osigma = \sigma(\ocs)$.  We assume the population is sufficiently large
that the average spending rate is approximately independent of any one individual's choice
of spending rate.

The disutility to an individual is the sum of the cost of infection and the
accrued cost from social distancing. When an individual has perfect
information about their own state, they will stop expensive prevention behaviors as soon
as they are infected.
If costs are exponentially discounted at rate $h$, then an individual
becoming infected at time $t$ would accumulate a disutility
\begin{gather}
	\int_0^{t} e^{-h\tau} c_s(\tau) d\tau + C_i(t) e^{-ht}.
\end{gather}
If a person makes it to the end of the game without being infected,
they will accumulate a disutility
\begin{gather*}
	\int_0^{t_f} e^{-h \tau} c_s( \tau) d\tau.
\end{gather*}
But each of these possible outcomes is uncertain.
If $p_s(t)$ is the probability of still being susceptible at time $t$
and $p_i(t)$ is the probability of being infected at time $t$,
then the expected disutility $D$ from $t=0$ to $t=t_f$ is
\begin{gather}
	D =
	\int_{0}^{t_f} \left( \int_{0}^{t} e^{-h\tau} c_s(\tau) d\tau + C_i(t) e^{-ht} \right) \dot{p}_i(t) dt
	+
	p_s(t_f) \int_{0}^{t_f} e^{-ht} c_s(t) dt
\end{gather}
Under System~\eqref{eq:PopEpi}, $p_i = 1-p_s$ and the hazard rate for this susceptible
individual to become infected is 
\begin{gather}
	\label{eq:ps1}
	-\dot{p}_s / p_s = \sigma(c_s(t)) \beta I
\end{gather}
with $p_s(0)=1$.  Using integration-by-parts, one can then show
\begin{equation}
	\label{eq:D1}
	D = \int_0^{t_f} e^{-ht} [c_s(t)+ \sigma(c_s(t)) \beta I(t) C_i(t) ] p_s(t) \,dt.
\end{equation}
The disutility $D$ depends on the individual's social-distancing investments $c_s(t)$,
and the course of the epidemic $I(t)$, which depends
on the average social-distancing costs $\ocs(t)$.  Thus,
$D(c_s, \ocs)$ as defined by Eq.~\eqref{eq:D1} under Eqs.~\eqref{eq:ps1} and 
\eqref{eq:PopEpi} defines a set of dynamic population games parameterized by the
values of $t_f$, $\sigma$, $C_i$, $h$, $\beta$, and $N$.

In this paper, we will restrict ourselves 
to an important special case of the SI social-distancing game where there is no discounting ($h=0$),
the cost-of infection $C_i$ is constant over time, and social distancing
has the form of a threshold-linear response
\begin{gather}
	\label{eq:maxsigma}
	\sigma(z) = \max( 0, 1-m z) = (1-mz)^+,
\end{gather}
where $m$ is the efficiency of spending. Eq.~\eqref{eq:maxsigma} implies
a rolling investment at rate of $1/m$ is enough to surely prevent infection.
Eq.~\eqref{eq:maxsigma}'s choice of $\sigma$ may seem arbitrary and
overly-restrictive, but its (weak) convexity is enough to weakly satisfy the economic
principle of diminishing returns (that investments become less efficient as
they grow because you do the easy things first). In this model, true diminishing returns
manifest only when $\sigma$ is strictly convex, but this more general setting makes
the game analysis much less tractable.
The threshold-linear form that we adopt allows us to exploit
linear-programming principles that lead to an analytically integrable system, and this proves to
be the linchpin for proving uniqueness. As we will discuss further below, while we are working
in a particularly special and convenient setting, the argument that we develop for uniqueness
can be generalized in a direction that should ultimately allow the establishment of uniqueness in 
the general SI social distancing game.


By a judicious choice of units for population, time, and money, we will
lose no generality in selecting $C_i=\beta=N=1$.  Thus, costs will
be measured relative to the cost of infection, $S$ and $I$ will represent
fractions of the total population.  In dimensionless time, and in the absence of intervention, half
the population is infected in \( \ln(1/I_0 - 1) \approx -2.3 \log_{10} I_0 \) time units.
(So if initially $1/1000$'th of the population is infected, then half the
population will be infected in about 7 time steps and only $0.1$\%
of the population will be uninfected after 14 time steps.  Around the time of
peak transmission, about 25\% of the population becomes infected in one timestep. )
We further simplify our notation, taking $p(t):= p_s(t)$
and $c := c_s/ C_i$.
And so finally, under our hypotheses of no discounting, constant infection costs, a threshold-linear
response feedback, and a perfect-vaccination terminal condition,
our dynamic SI social-distancing population game (\SISDG) has the dimensionless form
\begingroup
\begin{subequations}
	\label{eq:SISDG}
\begin{align}
	\label{eq:SISDGd}
	\tag{\SISDG-a}
	D(c ,\overline{c}) &= \int_0^{t_f} \left[c   + (1-mc)^+ I\right] p \,dt.
	\\
	p(0)=1, \quad \dot{p} &= - (1-mc)^+ I p, 
	\tag{\SISDG-b}
	\\ 
	I(0)=I_0, \quad \dot{I} &=  (1-m \oc)^+ I(1-I),\label{eq:SISDGa}
	\tag{\SISDG-c}
\end{align}
\end{subequations}
\endgroup
parameterized by the game-duration $t_f$, the efficiency-of-distancing $m$,
and the initial infected fraction $I_0$.

\begin{table}[t]
	\caption{Glossary of important symbols for \SISDG.
	We have selected units such that the total
	population size, the mass-action transmission
	rate, and the cost-of-infection are all $1$.
	\label{tbl:xs}
	}
	\begin{tabular}{cl}
	Symbol & Meaning
	\\
	\hline
		$t$ & Date coordinate
		\\ $S(t)$ & Fraction of population susceptibles
		\\ $I(t)$ & Fraction of population infected
		\\ $V(t)$ & Shadow value of susceptibility
		\\ $p(t)$ & Probability that an individual is susceptible at $t$
		\\ $c(t)$ & individual running cost of prevention
		\\ $\overline{c}(t)$ & population average running cost of prevention
		\\ $\sigma$ & relative susceptibility depending on running cost ($0 \leq \sigma \leq 1$)
		\\ $m$ & linear efficiency of running costs ($0<m < \infty$)
		\\ $t_f$ & Wait time until vaccine becomes available ($0< t_f < \infty$)
		\\ $I_0$ & Initial fraction of the population infected ($0 \leq I_0 \leq 1$)
	\end{tabular}
\end{table}

\section{Preliminary analysis}

Before proceeding, let us provide some context by first
commenting on some special limiting/boundary cases of the \SISDG.
First, observe that prevalence is increasing and bounded
($\dot{I} \geq 0$, and for $t\geq 0$, $0 < I_0 \leq I(t) <1$),
and strictly-increasing away from the extremes.
In a mean-field SI model, there is no stochastic extinction, no matter how small a faction of
the population $I_0$ is.
The probability of susceptibility at time $t$,
$p(t) \geq\exp\left(-\int_0^t I(t; I_0) dt \right) \geq \exp(-t) >0$. Because 
rolling costs are never negative ($c\geq 0$ and $c + (1-mc)^+ I \geq \min(1/m, I_0) > 0$),
the epidemic will always be expected to impose a
price on each individual ($D(c,\oc) \geq  \min(1/m, I_0)
(1-e^{-t_f}) > 0$), no matter their actions. 
On the other hand, a person
need never do worse than assured infection in the \SISDG.
(since $D(0,\oc) \leq 1-e^{-t_f} < 1$).
If the rolling cost of prevention exactly offsets the expected rolling cost of
infection ($m=1$), social distancing provides no benefits.  For all
less-efficient scenarios ($m < 1$), the upper-bound on $I(t)$ implies
rolling-costs of social distancing outweigh the expected rolling-costs of the
infection, and consequently that nobody should ever social-distance.

We are interested in finding strategic equilibria of this population game and
characterizing their properties.  We would like to find strategies that can not
be improved on when universally adopted but are themselves improvements over
adopted alternatives.  Said differently, we would like to find a strategy $c^*$
that satisfies the best-response condition
\begin{gather}
	\label{eq:nashcondition}
	c^* \in \argmin_{c} D(c,c^*)
\end{gather}
and the invasion potential condition
\begin{gather}
	\label{eq:invasion}
	c^* \in \argmax_{ \oc } \quad D(c^*,\oc) - D(\oc,\oc).
\end{gather}
A strategy that satisfies \eqref{eq:nashcondition} are called Nash equilibria, 
and a strategy that satisfies both \eqref{eq:nashcondition} and \eqref{eq:invasion}
is called an ESS (evolutionarily-stable strategy).  While there are 
ensuring the existing of Nash equilibria under relatively general conditions,
many games lack ESS strategies.  Additionally, the invasion potential condition
\eqref{eq:invasion} is difficult to analyze in many situations.  The same holds
here.  We are unable to analyze the invasion potential condition over the
general strategy space, and thus focus on proving uniqueness of those strategic
equilibria which satisfy the Nash condition, we do obtain insight by first
considering a useful subset of possible strategies.

\section{Delay strategies}

A natural first approach to the study of the SI social-distancing game
(\SISDG) is to try various strategies numerically.
With the benefits of hindsight, this is best done in a delay strategy-space.
Because we are modeling a permanent state-change, the risk-of-infection for a
susceptible individual is monotonically increasing until the terminal time. 
The threshold-linear form of $\sigma$ in Eq.~\eqref{eq:maxsigma} suggests
there may be a tipping point, before which the risk is too small relative
to the cost of distancing, and after which it is efficient to distance
maximally (a bang-bang form).  
Let us define a new coordinate variable $x \in [0,t_f]$ to specify
the duration of social-distancing at the end of the game, with $x=0$
implying no social distancing and $x=t_f$ implying social-distancing
for the full duration of the game.
A strategy can be represented in terms of the duration $x$ as 
a two-phase piecewise-constant function
$ \xi(t;x) = \operatorname{H}(t - t_f + x)/m$,
where $\operatorname{H}(t)$ is a Heaviside step function that is $0$ when $t$
is negative and $1$ when $t$ is positive.  
We only consider positive-duration games ($t_f > 0$), such that
each duration $x$ uniquely defines a strategy, and for any pair of delays
$(x,\ox) \in [0,t_f]\times[0,t_f]$, we integrate the \SISDG\xspace to
explicitly construct the delay-strategy SI social distancing population game's
(\DSISDG's) disutility 
\begin{subequations}
\begin{align}
	\D(x,\ox) := D(\xi(t;x),\xi(t;\ox))
	&= \int_{0}^{t_f-x} I p\, \mathrm{d}t
	+ \int_{t_f-x}^{t_f} p/m\, \mathrm{d}t 
	\tag{\DSISDG-a}
	\\
	=&
	\begin{cases}
		1 - q(x,\ox) \left(1-x/m \right) \quad \text{if $I_0 < 1$,} \label{eq:DSISDGb}
		\tag{\DSISDG-b}
		\\
		1 - e^{x - t_f} \left(1-x/m\right) \quad \text{if $I_0=1$,}
	\end{cases}
	\\
	\text{where} \quad
	q(x,\ox) :=&  \frac{1 - I(s)}{1-I_0} \exp\left[ \left(s - t_f + x \right) I(s) \right], \label{eq:DSISDGc}
	\tag{\DSISDG-c}
	\\
	s :=&  t_f-\max(x,\ox),
	\tag{\DSISDG-d}
	\\
	\label{eq:Ifunc}
	\text{and} \quad
	I(u) :=& I_0/\left[ I_0 + (1-I_0)\exp({-u})\right].
	\tag{\DSISDG-e}
\end{align}
Here, $q(x,\ox)$ is the probability an individual is still susceptible when
they start social-distancing.
\end{subequations}
These closed-forms for $\D(x,\ox)$ allow us to extract insights to the game-play of
the \DSISDG\xspace that are inaccessible in the general cases of the \SISDG.
In the economics literature, this is sometimes categorized as an optimal-stopping
game (where stopping here would be to stop taking risks).

\subsection{Equilibria}

\begin{figure}
	\centering \includegraphics[width=\textwidth]{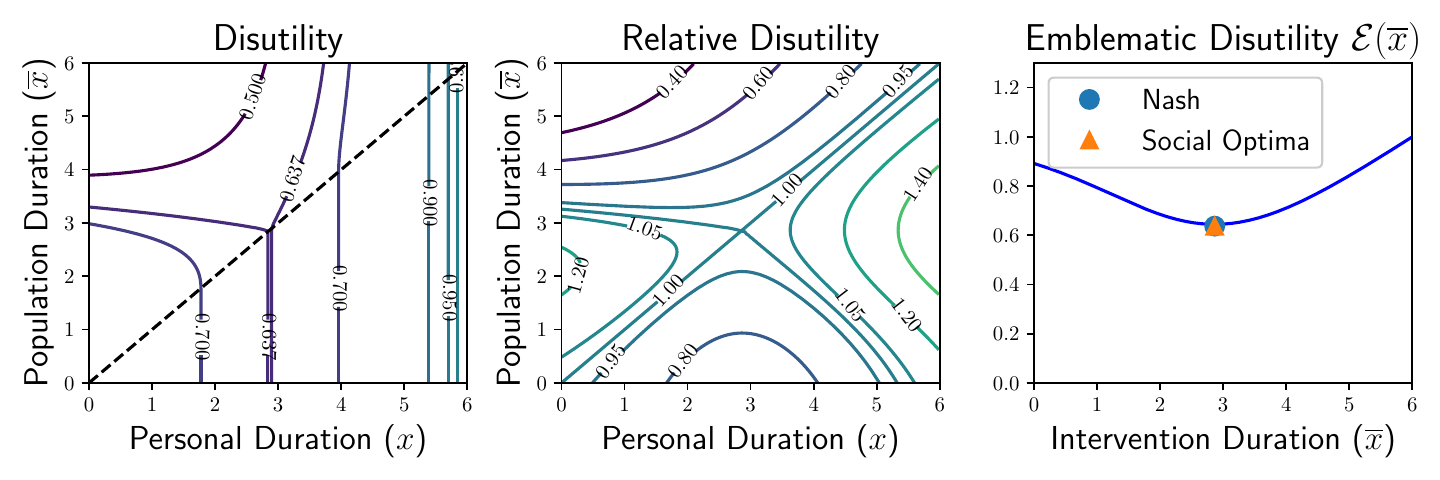} 
	\caption[Two-Phase Strategy Stability]{ \label{fig:NonConstStrat}
	The delay disutility surface
	$\D(x,\ox)$ (left), the relative delay disutility $\hat{\D}(x,\ox)$ (center), 
	and the emblematic disutility $\mathcal{E}(\ox)$ (right) when $t_f = 6$, $m=6$, and initial
	condition $I_0 = 0.02$. The strategy at the saddle point ($x^* \approx 2.87$)
	of the center plot is the Nash equilibrium among all off-on two-phase strategies
	and is also the minimizer of the emblematic disutility (right), so is also the social optimal.
	}
\end{figure}

A Nash equilibrium $x^*$ of the \DSISDG\xspace is a strategy that is
always a best-reply to itself, (a solution of Eq.~\eqref{eq:nashcondition} for
$x \in [0,t_f]$), and a strict Nash equilibrium if it is the only best-reply.
Even though the strategies being considered are step functions at their root,
we can prove that the delay disutility $\D(x,\ox)$ itself is
positive and continuous on $[0,t_f]\times [0,t_f]$, including near the line
$x=\ox$ (See~Fig.~\ref{fig:NonConstStrat}).  We can further show that the
Jacobian of $\D(x,\ox)$ is defined uniquely everywhere within this domain (Appendix~\ref{secA3}),
allowing us to directly apply standard calculus techniques to calculate the
equilibria.
The first result is that for all parameter tuples $(m,I_0,t_f)$,
there exists one and only one equilibrium point, $x^*$
(Appendix~\ref{secA2}, Theorem~\ref{thm:ESS}, Eq.~\eqref{eq:NashCondition}).
We can further show that $x^*$ is a global Nash equilibrium with invasion
potential, so is an ESS (see Appendix~\ref{secA2}, Theorem~\ref{thm:ESS}).
For games of constant risk ($I_0 = 1$), $x^* = \min( t_f, \max(m-1,0))$.
For short games where $t_f < m$ and with high initial risk ($I_0
> 1/( m - t_f)$), the equilibrium is to always distance ($x^*=t_f$).
For short games with low initial risk ($I_0 < 1/(1+(m-1)e^{t_f})$),
the equilibrium is to never distance ($x^*=0$).
For intermediate values of the initial condition $I_0$,
this equilibrium will satisfy 
\begin{gather}
	e^{x} =
	e^{t_f} \left( \frac{I_0}{1-I_0} \right) ( m - 1 - x ).
\end{gather}
This transcendental equation has a single real solution $x^*$ with a closed-form representation in terms of
the Lambert-Corless $\operatorname{W}$-function \citep{bib:Corless96},
%
\begin{gather}
	\label{eq:DSISDG:nash}
	x^*(m,I_0,t_f) := \begin{cases}
		0 & \text{if $I_0 (1+(m-1)e^{t_f}) < 1$},
		\\
		t_f & \text{if $I_0 (m-t_f) > 1$},
		\\
		m - 1 - \operatorname{W}\left( \left(\frac{1}{I_0} - 1\right) e^{m-1 - t_f} \right)
		& \text{otherwise.}
	\end{cases}
\end{gather}
Asymptotically, for long games ($t_f \rightarrow \infty$),
\begin{gather}
	x^* = m-1 - (1/I_0-1) e^{m-1-t_f} + o(e^{-t_f}).
\end{gather}
On the other hand, when the distancing is brief and $I_0 < 1/m$,
\begin{gather}
	\label{eq:nashline}
	x^* =
	(1 - 1/m) \left[t_f - \ln{(1/I_0-1)} + \ln(m-1)\right]^+ + O(t_f^2).
\end{gather}
Given our exact solution, 
we omit the construction of the interpolating corner-layer
in the intermediate region of $t_f$ where both approximations
break down, as it supplies no additional conceptual context.
The parameter-dependences of $x^*$ are illustrated in Figure~\ref{fig:gamespacenash}.
A dimensional example is shown in Figure~\ref{fig:dimen}.

\begin{figure}
	\begin{center}
	\includegraphics[keepaspectratio,width=\linewidth]{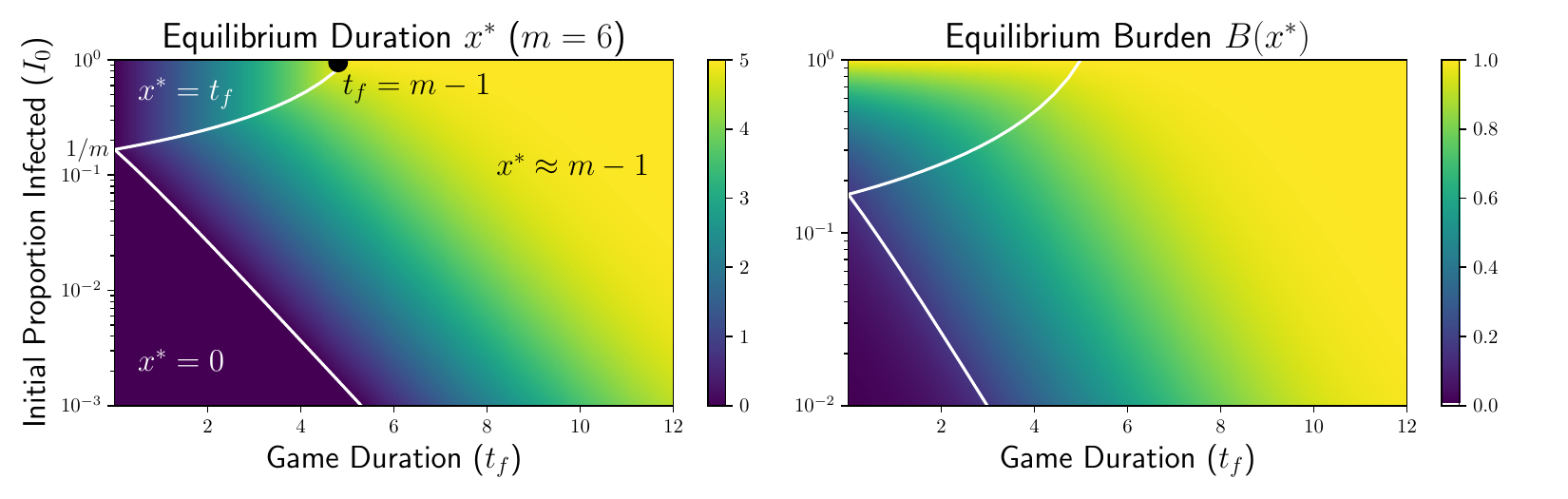}
	\end{center}
	\caption{ \label{fig:gamespacenash}
	Colorized contour images of the (left) equilibria duration of social
	distancing $x^*$ of the \DSISDG\xspace as
	a function of the game-duration ($t_f$) and the initial proportion
	infected ($I_0$) depends when the linear distancing efficiency $m=6$.
	As games get longer, the duration of distancing increases from $0$
	to about $5$.  The white line divides the surface into three
	regions depending on whether no distancing is used ($x^*=0$), distancing
	is used for the full game-duration ($x^*=t_f$), or some intermediate
	amount is used.  The equilibrium per capita burden
	(right, Eq.~\eqref{eq:burden}) increases monotonically as both the game duration
	and the initial case count increase.  When $I_0 \approx 1$, almost everybody
	gets sick and the burden is near $1$ also, even when the game is short
	and susceptibles social distance the whole time.
	}
\end{figure}

\begin{figure}
	\begin{center}
	\includegraphics[keepaspectratio,width=3in]{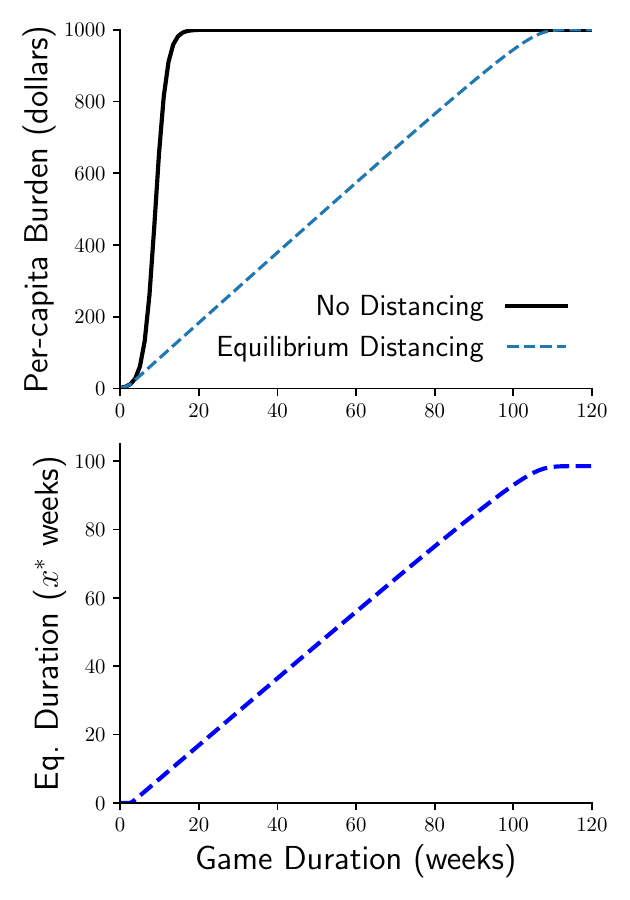}
	\end{center}
	\caption{
	\label{fig:dimen}
	A dimensionalized example of equilibrium (bottom) and corresponding burden (top,
	Eq.~\eqref{eq:burden}).  For a community of 10,000 people the epidemic is detected
	when the first $\hat{I}_0 = 20$ are infected.  The epidemic is initially
	doubling in size each week ($=\ln 2 / \hat{\beta} \hat{N}$), infected people
	expect to lose \$1,000 ($=C_i$) from the infection, but are not willing to pay
	more than \$10 ($=1/\hat{m}$) a week in social distancing costs, then
	social-distancing can only reduce the burden significantly if the vaccine is
	rolled out in less than 100 weeks.
	}
\end{figure}

\subsection{Interpretation}

Further analysis is needed to interpret the implications of
$x^*(m,I_0,t_f)$ from Eq.~\eqref{eq:DSISDG:nash}.  We construct the
emblematic delay disutility 
\begin{gather} \label{eq:emblem}
	\mathcal{E}(\ox) := \D(\ox,\ox) = 1 - \frac{(1-\ox/m)}{ 1 - I_0 + \exp(t_f-\ox)}
\end{gather}
which tells us the typical accrued costs to an initially-susceptible
individual playing the population-strategy $\ox$.  The relative
delay disutility
\begin{gather}
	\hat{\D}(x,\ox) := \D(x,\ox)/\mathcal{E}(\ox)
\end{gather}
tells us how badly off a strategy $x$ performs relative to the
population strategy $\ox$.  For an example, see
Figure~\ref{fig:NonConstStrat}.

Our first conceptual result is that the 
\DSISDG\xspace never suffers from free-riding (Appendix~\ref{secA2}) --
the strategy that optimized the populations welfare also
optimizes the individual's well-being in the sense that
\begin{gather}
	x^* = \argmin_{\ox \in [0,t_f]} \mathcal{E}(\ox).
\end{gather}
At equilibrium, the emblematic disutility has the convenient form
\begin{gather}
\mathcal{E}(x^*) =
	\begin{cases}
		(1-I(t_f))/(1-I_0), &\text{if $x^*=0$},
		\\
		(1+x^* - m I_0)/m (1-I_0), &\text{if $x^* \in (0,t_f)$},
		\\
		t_f/m, & \text{if $x^* = t_f$}.
	\end{cases}
\end{gather}

To see how much social-distancing would actually help (see~Fig.~\ref{fig:gamespacenash}-\ref{fig:improvement}), we define
the epidemic burden on the
population as the sum of the accrued costs of infection and
prevention over the whole population, 
\begin{align}
	\label{eq:burden}
	B(\ox) :=& I_0 + (1-I_0) \mathcal{E}(\ox)
		\\
		=& I(t_f-\ox) + \frac{\ox}{m} \left( 1 -  I(t_f-\ox) \right)
\end{align}
and the equilibrium improvement-over-indifference
\begin{gather}
	\label{eq:improve}
	\Delta B^* := B(0) - B(x^*)
\end{gather}
tells us how much equilibrium social-distancing reduces the burden of
the epidemic relative to the case where no preventative actions are
taken ($\ox=0$).
Figure~\ref{fig:twoburdens} gives two example intensity-plots of
the burden for $I_0=1$ and $I_0 = 10^{-4}$.

Numerically, we can see that the impact of social-distancing on
overall public-health is concentrated at intermediate game-durations
(Figure~\ref{fig:improvement} and \ref{fig:meffect}) -- roughly $\ln m I_0$ to $m + \ln mI_0$.
The improvement-over-indifference $\Delta B^*(m,I_0,t_f)$ defined
by Eq.~\eqref{eq:improve} is maximized when the game duration
\begin{gather}
	\label{eq:tpeak}
	t_f^{@}(m,I_0)
	:\approx \ln\left( \, (m-1)(1/I_0 - 1) \, \right),
\end{gather}
which increases as both the initial case-fraction and the efficiency
of distancing increase.  This is derived through differentiation
of \eqref{eq:improve} after substitution of \eqref{eq:nashline}.  
As $I_0 \rightarrow 0$,  we can show there is an upper bound on the improvement-over-indifference
that only depends on the efficiency-of-distancing (see Fig.~\ref{fig:meffect}), 
\begin{align}
	\max_{t_f} \Delta B^* \approx
	1 - \frac{2}{m} - \frac{ 2 (m - 1) \ln (m - 1 )}{m^2}
	\approx
	1 - \frac{2(1+\ln m)}{m}.
\end{align}
For games much shorter than $t_f^{@}$, risk-of infection never rises
high enough for social-distancing to have much of an impact.  For games much
longer than $t_f^{@}$, the equilibrium's distancing-induced reduction in expected total costs
is very small.  Thus, social distancing has a Goldilocks window of the outbreak response time, depending
on when the outbreak is detected and the efficiency of social distancing.

\begin{figure}
	\begin{center}
	\includegraphics[keepaspectratio,width=5.5in]{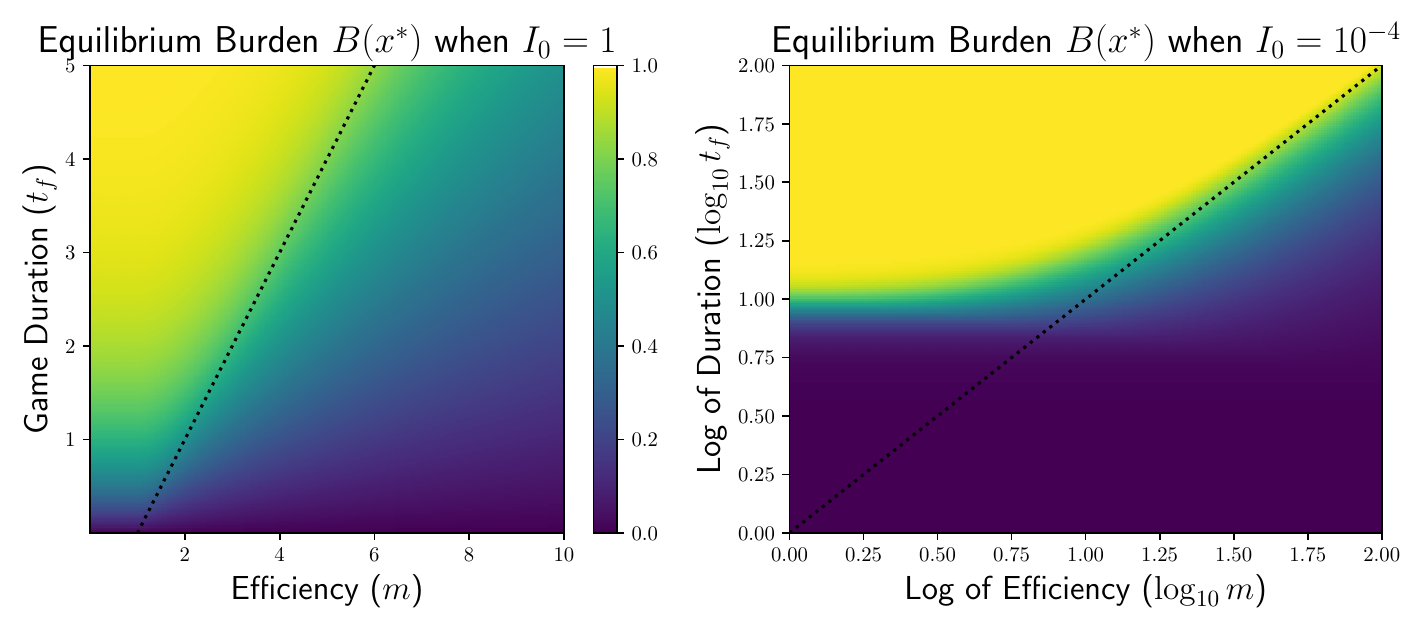}
	\end{center}
	\caption{
	\label{fig:twoburdens}
	Epidemic burden as functions of $t_f$ and $m$ for given values of $I_0$.  Dotted
	black line represents $t_f = m-1$.
	The case of $I_0=1$ (left) is
	gameless, as infection-risk is constant and independent of player
	actions.  We see the burden increases monotonically as duration $t_f$
	get longer but decreases monotonically as social-distancing is 
	made more efficient.  In the case of $I_0=10^{-4}$ (right), burden is small
	for short games because infection risk never becomes
	substantial, but for long games, the burden only decreases
	significantly once $m/t_f \approx 1$.
	}
\end{figure}

\begin{figure}
	\begin{center}
		\includegraphics[keepaspectratio,width=\linewidth]{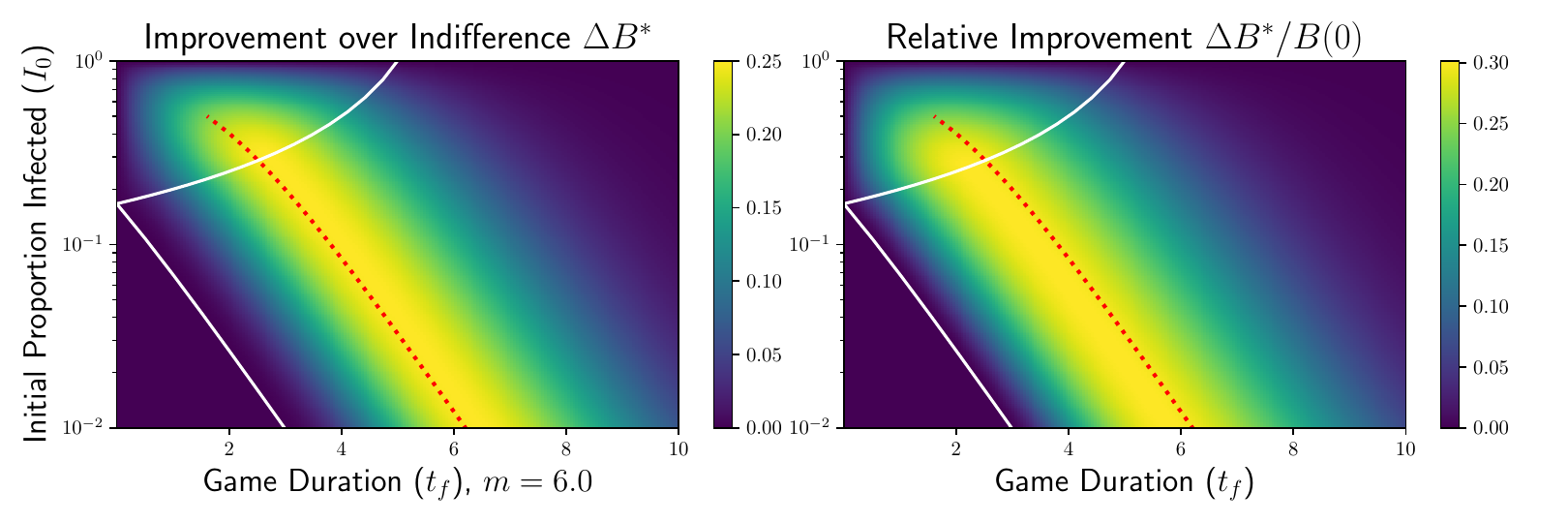}
	\end{center}
	\caption{ \label{fig:improvement}
	Colorized contour image of 
	the reduction in burden by social distancing ($\Delta B^*$, Eq.~\eqref{eq:improve})
	(left),
	and the relative reduction in burden (right) when $m=6$.
	The plots clearly show that for this epidemic
	model, the public-health value of social-distancing is concentrated
	at intermediate durations of time between detection and mass-vaccination.
	The dashed red line is our approximation of the time $t_f^@(m,I_0)$
	(Eq.~\eqref{eq:tpeak}) when the improvement-over-indifference
	$\Delta B^*$ is maximal.  
	}
\end{figure}

\begin{figure}
	\begin{center}
		\includegraphics[keepaspectratio,width=3in]{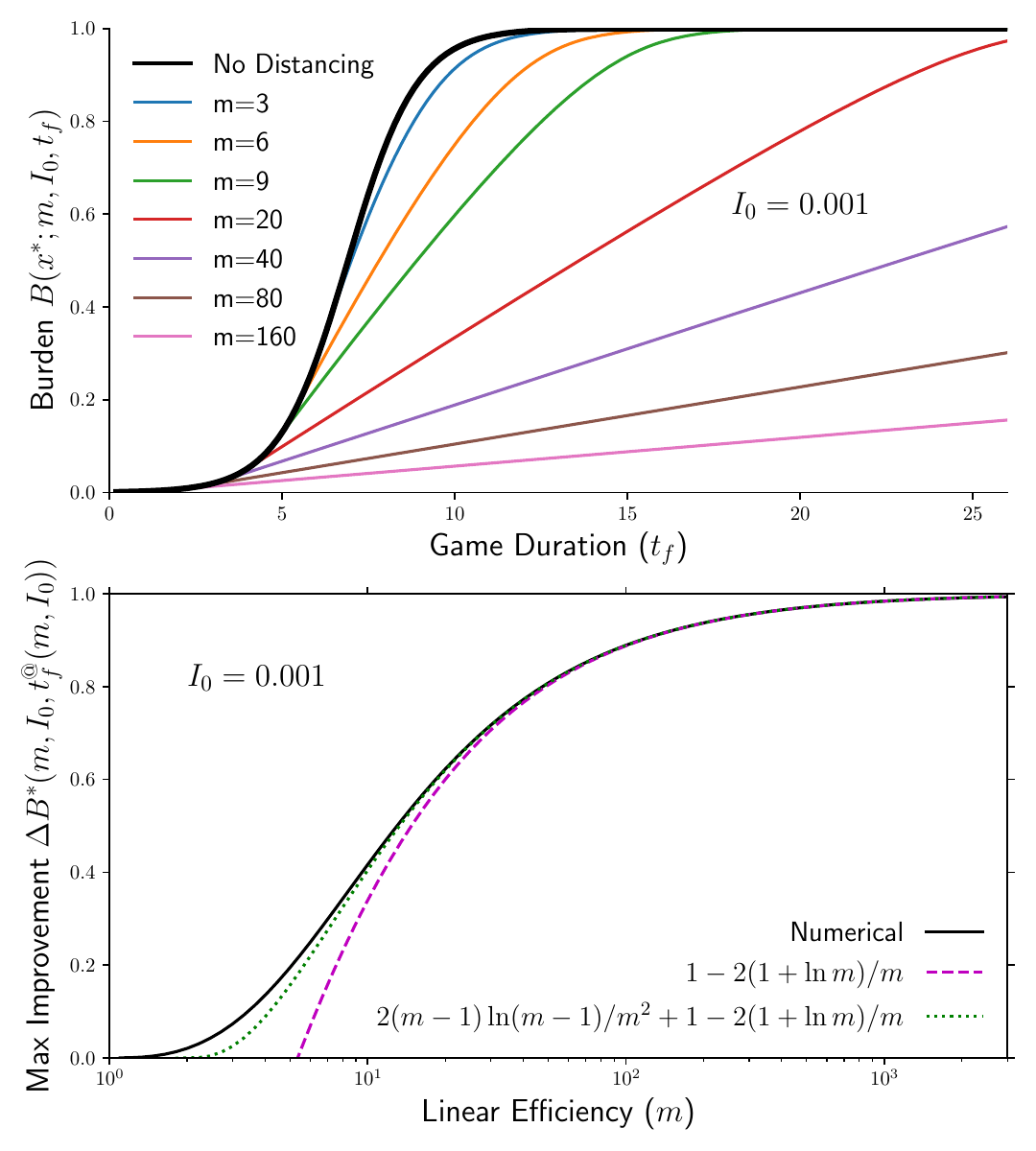}
	\end{center}
	\caption{
	\label{fig:meffect}
	The burden of epidemic under equilibrium distancing
	is always an increasing function of the game duration (top).
	Increasing the efficiency of distancing decreases the rate
	of burden increase, but for game durations that are very large
	or very small, equilibrium social distancing can not meaningfully reduce the 
	burden.  The Goldilocks window for distancing reducing burden gets larger
	roughly in proportion to the efficiency of distancing.
	The maximum improvement in burden (bottom) increases slowly
	with $m$.
	Parameter: $I_0 = 10^{-3}$.
	}
\end{figure}

\section{General Analysis}

To generalize our analysis, we will allow individuals to adopt piecewise-continuous strategies $c:[0,t_f] \to [0,\infty)$. 
In this setting, Eq.~\eqref{eq:nashcondition} becomes
\begin{equation}
\label{eq:NashCond}
	c^* \in \argmin_{c\,\in\, [0,t_f] \rightarrow [0,\infty)} D(c,c^*). 
\end{equation}
Our present goal is to establish that each game described by
\SISDG\xspace has just one Nash equilibrium. 

An additional condition we require is that $c^*$ is subgame-perfect \citep{bib:Fudenberg1991}. In
practice, this condition simplifies the process of determining a best-response,
and implies that all individuals are acting in their own
best-interest at all times between the start and end of the game.

To analyze the \SISDG\xspace for its more general strategy space, 
we adopt the standard approach of fixing the population behavior $\overline{c}$ and determining the 
best response $c_B(\overline{c})$ using a Hamiltonian approach, 
specifically through the application of Pontryagin's Maximum Principle \citep{bib:Clarke2013}. 

We define the Hamiltonian as 
\begin{equation}
	\mathscr{H} := - [c + (1-mc)^+I]p - (1-mc)^+ I pV + k (1 - m \overline{c})^+ I(1-I).
\end{equation}
Applying the maximum principle, we obtain the present time personal state adjoint equation
\begin{equation}
\dot{V}  = (1-mc)^+ I (V+1) + c.
\end{equation}
with terminal condition $V(t_f) = 0$. Because we treat the population behavior
$\overline{c}$ as fixed, we can safely ignore $k$ and its corresponding
equation. This is further justified by the following: $V(t)$ can be interpreted
as the present time cost of being susceptible \citep{bib:Clark2010}, and
applying integration by parts to the \SISDG, we obtain
\begin{equation}
\label{eq:DV}
D(c, \overline{c}) = -V(0). 
\end{equation}
Therefore, the initial cost of being susceptible at time 0 is equivalent (up to the sign) to the disutility of playing $c$ against $\overline{c}$. This relation will be key to our analysis, as it will allow us to connect the boundary value problem obtained through the optimization process to an initial value problem that we can construct solutions for. 

To determine $c_B(\overline{c})$, we utilize an argument typically deployed in dynamic programming. 
In order to minimize the disutility, 
Eq.~\eqref{eq:DV} states that we need to minimize the magnitude of $-V(0)$. Starting at $V(t_f) = 0$, $c_B(\overline{c})$ is
constructed in reverse time to lead to the minimal increase in $-V(t)$ at each $t \in [0,t_f]$. Therefore, we must have 
\begin{equation}
\label{eq:bestresponse}
c_B(\overline{c}) \in \argmin_{c \in [0,\frac{1}{m}]} \{ (1-mc)^+ I(V+1) + c\}.
\end{equation}
Constructing the best response in this fashion removes explicit dependence upon $\overline{c}$, only depending on the population behavior through the state $I$. 
This dependency of $c^*$ on only the state variables $I$ and $V$ implies that both the personal and population 
strategies are specified at each point in time through Eq.~\eqref{eq:bestresponse}. 
This naturally constructs a Nash equilibrium because the two strategies match and the personal strategy is by construction a 
best response. Also, because this strategy is optimal at each time point by construction, it is also subgame perfect. 

Taking $c^*$ to be the Nash strategy which satisfies Eq.~\eqref{eq:bestresponse} for each $(I,V)$, we obtain the 
Filippov boundary-value problem \citep{bib:Filippov}
\begin{subequations}
\label{eq:OldSystem}
\begin{align}
I(0)=I_0, \quad
\dot{I} &= (1-mc^*)^+ I(1-I), \label{eq:OldSystema}\\
\dot{V} &= (1-mc^*)^+ I(V+1) + c^*, \quad V(t_f)=0, \label{eq:OldSystemb}\\
c^*(I,V) &\in \label{eq:OldSystemc}
\begin{cases}
\{0\} & \text{if} \; m I \big (V+ 1 \big) < 1, \\
[0,1/m] & \text{if} \; m I \big (V+ 1 \big) = 1, \\
\{1/m\} & \text{else}.
\end{cases}
\end{align}
\end{subequations}
with boundary conditions $I(0) = I_0$ and $V(t_f) = 0$. It should be noted that the ``bang-bang'' form of $c^*$ given by 
Eq~\ref{eq:OldSystemc} is entirely due to the linearity of $(1-mc)^+$ for $c \in [0,\frac{1}{m}]$. If $\sigma$ were instead
strictly convex, determining $c^*$ from Eq~\ref{eq:bestresponse} would generally give an equilibrium strategy that 
changes continuously instead of having distinct regions of constancy. This continuous variation in $c^*$ is one of
the primary reasons why the general setting is significantly more difficult than the setting we are working in. As will be 
demonstrated below, our argument leans heavily on these regions of constant $c^*$. 

The feedback-response form of $c^*(I,V)$ allows us to study the properties of solutions to the \SISDG\xspace
from the phase-plane structure of System~\eqref{eq:OldSystem}.
To aid us, we restrict our study this system to $(I,V) \in (0,1] \times ( -\infty, \infty)$. 
To perform our analysis, we need to translate this boundary value problem into an initial value problem 
with initial condition $(I_0, V_0)$. 
To aid us, the following lemma will restrict the set of $V_0$ which give 
solutions corresponding to a game. 

\begin{lemma}
\label{lemma:divergence}
Fix $I_0 \in (0,1]$ and $V_0 \in (-\infty, \infty)$. Suppose $V(t)$ solves Eq.~\eqref{eq:OldSystemb}. 
If there exists a $t \in [0,\infty)$ with $V(t) = 0$, then $V_0 \in [-1,0]$.
\end{lemma}

\begin{proof}
We will prove this by showing that  $V_0$ outside $[-1,0]$ do not give $V(t)$ that ever satisfy the terminal condition. 
When $V < -1$, $c^* = 0$ because $I(V+1) < 0 < \frac{1}{m}$. In this case, we have 
\[
\dot{V} = I(V+1) < 0.
\] 
This implies that $V(t)$ is decreasing for all $V < -1$, and thus for no $t \in [0,\infty)$ can $V(t) =0$. 

When $V > 0$, for all $c^* \in \Big[0,\frac{1}{m} \Big]$,
\begin{gather}
\dot{V} = (1- m c)^+ I (V+1) + c > 0. 
\end{gather}
This implies that $V(t)$ is increasing for all $V > 0$, and thus for no $t \in [0,\infty)$ can $V(t) =0$. 
\end{proof}

Lemma~\ref{lemma:divergence} 
restricts the domain where we consider solution trajectories to the system to $(I,V) \in (0,1]\times[-1,0]$. We will
exclude the disease free state $I = 0$ from further consideration since it is trivially $c^* = 0$ for any duration. 

Our approach to establish uniqueness of $c^*$ is to 
integrate System~\eqref{eq:OldSystem} forward in time from the initial line defined by $I(0) = I_0$ and 
determine the terminal time function $t_f(V_0;I_0)$ such that $V(t_f(V_0;I_0)) = 0$. 
Uniqueness for the game follows by demonstrating a 
bijective relationship between the game duration $t_f$ and the initial adjoint value $V_0$. 
With such a bijection, for each game duration $t_f$ there exists a 
unique initial $V_0$ such that $V(t)$ satisfies the necessary condition with $V(t_f) = 0$. Since the optimal strategy is defined
in feedback-response form, the solution trajectory of the IVP from $(I_0,V_0)$ under $c^*(I,V)$ to the terminal time $t_f$ 
uniquely specifies $c^*$. 
This approach is mentioned by \citet{bib:Reluga2013} for the finite-horizon case, but the complexities of that system 
make it more difficult. 

Before we begin our analysis, it is helpful to transform $V$ into a form that is easier to use. 
We see in Eq.~\eqref{eq:OldSystemc} that the level sets of the decision potential function
\begin{equation}
\label{eq:phidef}
\Phi(I,V) := I \big ( V + 1 \big)
\end{equation}
determine the value of $c^*$. 
For the following arguments we will fix $I_0>0$, and since prevalence is monotonically increasing,
we can replace $V$ with $\Phi$ as a state variable
since for each $I>0$ there is a continuous bijection between $\Phi$ and $V$ given by
\[
V = \frac{\Phi}{I} -1.
\]
We depict the result of this transformation in Figure~\ref{fig:CoordTrans}.

\begin{figure}
\centering
\scalebox{1.25}{
\includegraphics{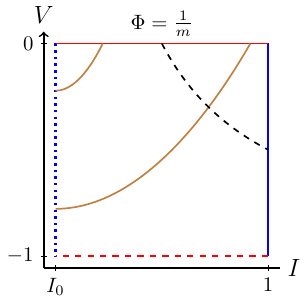}
~
\includegraphics{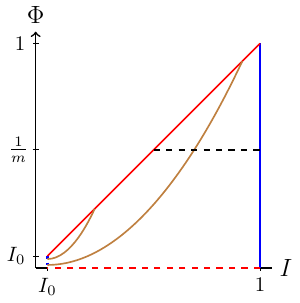}
}
\caption[Coordinate Transformation]{\label{fig:CoordTrans}The enclosed region is admissible for solution trajectories and depicts how they are mapped between $(I,V)$ space and $(I, \Phi)$ through $\Phi = I (V + 1)$. The dashed line is the transition line $\Phi = \frac{1}{m}$ and the two brown curves represent how trajectories are mapped after the coordinate transformation.}
\end{figure}

Although $\Phi$ as defined is a geometric quantity, 
the dynamics of $I$ and $V$ imply a dynamic for $\Phi$.
Using the chain rule, $\Phi(t) := \Phi(I(t),V(t))$ must satisfy
\begin{equation}
\frac{d}{dt} \Phi = \nabla \Phi \cdot (\dot{I}, \dot{V}) = (1-mc)^+ \Phi + cI.
\end{equation}
The initial value $(I(0), V(0)) = (I_0,V_0)$ gives $\Phi(0) = \Phi_0 = I_0(V_0 + 1)$. 
One of the key geometric changes resulting from this
transformation is that the terminal condition $V=0$ becomes the line $\Phi = I$.  
After eliminating $V$ this transformation results in the equivalent IVP that will be the focus of our study.

\begin{problem}
\label{nc:ch2}
If $c^*(t)$ is a subgame perfect Nash equilibrium to a \SISDG\xspace of duration $t_f \in[0, \infty)$ 
with efficiency $m>0$,  then 
\begin{subequations}
	\label{eq:IPhiopt}
	\begin{equation}
	c^*(\Phi) \in
	\begin{cases}
	\{0\} & \text{if } \Phi < 1/m, \\
	[0,1/m] & \text{if } \; \Phi = 1/m, \\
	\{1/m\} & \text{else}.
	\end{cases}
	\end{equation}
	where $\Phi(t_f) = I(t_f)$, $I(0) = I_0$, and
	\begin{align}
	\dot{I} &= (1-mc^*)^+ I (1-I),\label{eq:IPhiopta}\\
	\dot{\Phi} &= (1-mc^*)^+  \Phi + c^* I. \label{eq:IPhioptb}
	\end{align}
\end{subequations}
\end{problem}
The primary feature of this coordinate transformation into decision potential form is that
the optimal strategy $c^*$ becomes solely dependent on $\Phi$, which makes phase-plane analysis of this system simpler. 
When $c^*$ is constant, we have explicit solutions to Eq.~\eqref{eq:IPhiopta} (Eq.~\eqref{eq:Ifunc}) and Eq.~\eqref{eq:IPhioptb} 
in the following lemma. 
\begin{lemma}
	If $I(t_0) \in (0,1]$, $\Phi(t_0) \in [0,1]$, and $c  \in \big[0, \frac{1}{m}\big]$ be fixed, and if $I(t)$ be given by 
	Eq.~\eqref{eq:Ifunc}, then the solution to Eq.~\eqref{eq:IPhioptb} is given by 
	\begin{equation}
		\label{eq:Phisol}
		\Phi(t) \exp{\big(-\sigma(c)(t-t_0) \big)} - \Phi(t_0) = \int_{t_0}^t c I(s)  \exp{\big( - (\sigma(c) + h)(s-t_0) \big)} \; ds
	\end{equation}
\end{lemma}
This follows from using an integrating factor, and when $c = 0$ or $\frac{1}{m}$, we will be able to integrate this into analytic functions which will greatly simplify our analysis. 

It is a standard result from \citet{bib:Filippov} that solutions exist for the initial value problem for System~\eqref{eq:IPhiopt} with
initial condition $(I_0, \Phi_0) \in (0,1]\times (0,I_0]$. Solutions to Filippov Systems are not generally unique, but solutions to
System~\eqref{eq:IPhiopt} are unique because $\dot{\Phi}>0$, particularly on the transition line $\Phi = \frac{1}{m}$. 
This is clear since $I_0,\Phi_0>0$ and $c^* \in [0,\frac{1}{m}]$, so by Eq~\ref{eq:IPhioptb}, $\dot{\Phi} > 0$ at the initial time.
From here, Eq~\ref{eq:Ifunc} demonstrates that $I(t) > 0$ for all $t \geq 0$, so $\dot{\Phi}$ can only ever become non-positive if 
$\Phi$ is, but this would  clearly require $\Phi_0 \leq 0$ which is excluded; therefore, $\dot{\Phi}>0$. 
Details on uniqueness of Filippov systems in general can be found in \citep{bib:Filippov}. 
We note in passing that the situation is more complicated in the presence of discounted
future returns ($h>0$) -- $\dot{\Phi}$ can sometimes change sign along the transition line and one must take care to establish
that the induced non-uniqueness of the Filippov System is irrelevant for the boundary conditions.
of System~\eqref{eq:IPhiopt}. These details will be addressed in a forthcoming sequel to this work. 

\section{Proof of Uniqueness}

We now present our proof of uniqueness of the Nash equilibrium $c^*$ for
the \SISDG.  $\Phi(t)$ is monotone increasing and can be integrated
to an analytic expression, which allows us to avoid many difficulties.
In particular, the optimal strategy can change at most once when $\Phi = \frac{1}{m}$,
and we denote that time as $\tau$. Also, because $\dot{\Phi} \geq
\min{\big(\Phi_0,\frac{I_0}{m} \big)} > 0$, and the game terminates at some $\Phi
\leq 1$, so System~\eqref{eq:IPhiopt} with any initial condition $(I_0, \Phi_0) \in
(0,1] \times (0,1]$ satisfies the terminal condition in finite-time. 

\begin{figure*}[t!]
\centering
\begin{subfigure}[t]{.49\textwidth}
\centering
\scalebox{1.25}{
\includegraphics{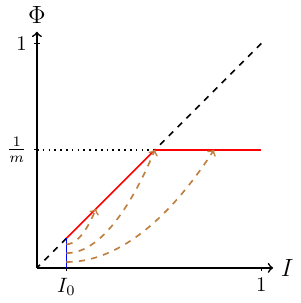}
}
\caption{\label{fig:TrajectoriesA} $c^* = 0$ region.}
\end{subfigure}
~
\begin{subfigure}[t]{.49\textwidth}
\centering
\scalebox{1.25}{
\includegraphics{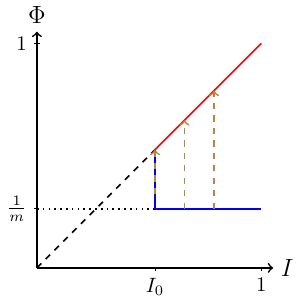}
}
\caption{\label{fig:TrajectoriesB} $c^* = \frac{1}{m}$ region.}
\end{subfigure}
\caption{ \label{fig:Trajectories}
In both graphics, the blue line is the set of potential initial points of a pure strategy trajectory, the dashed brown lines are hypothetical trajectories, and the red lines are the terminal points of a pure strategy trajectory. 
Figure~\ref{fig:TrajectoriesA} depicts $c^* = 0$ trajectories, which either terminate or transition when they
intersect  $\min{(I_0, \frac{1}{m})}$. 
Figure~\ref{fig:TrajectoriesB} depicts $c^* = \frac{1}{m}$ trajectories, which are either single-phase when 
$\Phi_0 \geq \frac{1}{m}$ or are continuations of two-phase strategies and so begin along the horizontal $\Phi = \frac{1}{m}$. 
}
\end{figure*}

The first step in our uniqueness proof will handle all $1 \geq I_0 > 0$ with $\Phi_0 < \frac{1}{m}$ in a single lemma. 
These trajectories either reach the terminal time surface first or they will intersect the transition surface,
the first case only being possible when $I_0 < \frac{1}{m}$. 
We can combine the transition and terminal surfaces into the line $\min{(I_0, \frac{1}{m})}$
and show that as $\Phi_0$ decreases with $I_0$ fixed, the time to intersect this line increases. This geometry is depicted in 
Figure~\ref{fig:TrajectoriesA}. 

\begin{lemma}
\label{lemma:SI1}
For $\Phi_0 < \min{(I_0, \frac{1}{m})}$, let $t^*$ be the time such that from $(I_0, \Phi_0)$ to $(I,I)$ when $c^* = 0$ is 
fixed, and $\tau$ as
the time at which $\Phi(\tau) = \frac{1}{m}$. Then $\min{(t^*,\tau)}$ is strictly monotone decreasing in $\Phi_0$, and thus
$\min{(t_f,\tau)}$ is strictly monotone decreasing in $\Phi_0$.
\end{lemma}

\begin{proof}
Fix $\Phi_0  < \min{(I_0, \frac{1}{m})}$. This specifies that $c^* = 0$, and we study the solutions from $(I_0,\Phi_0)$ until either $\Phi(t) = I(t)$, or $\Phi(t) = \frac{1}{m}$, whichever happens first. 

To begin, we will ignore the terminal condition and demonstrate that as $\Phi_0$ decreases, $\tau$ increases. 
When $c^* = 0$, we have for $0 \leq t \leq \tau$ that Eq.~\eqref{eq:Phisol} reduces to
\begin{equation}
\label{eq:simplephisol}
\Phi(t) = \Phi_0 e^{t}. 
\end{equation}
Coupling Eq.~\eqref{eq:simplephisol} with the definition of $\tau$, we obtain
\begin{equation}
\label{eq:Phi0Tau}
\tau = - \ln{(m \Phi_0)}.
\end{equation}
From Eq.~\eqref{eq:Phi0Tau}, we conclude that as $\Phi_0$ increases, $\tau$ decreases as desired. This also directly determines
 the duration of time spent playing $c^*$ in a two-phase strategy from the initial decision potential $\Phi_0$. 

The parallel to this result is showing $t^*$ decreases as $\Phi_0$ increases. We will show this generally for $c^*=0$.
The condition of intersection is $I(t^*) = \Phi(t^*)$, which in light of 
Eq.~\eqref{eq:Phisol} for $\Phi(t)$ and the solution for $I(t)$
given in Eq.~\eqref{eq:Ifunc},  can be written
\begin{equation}
\label{eq:Phi0tf}
\Phi_0 = \frac{ I_0}{(1-I_0)  + I_0 \exp{(t^*)}}.
\end{equation}
From this relation we see that as $\Phi_0$ decreases, $t^*$ must increase to retain equality. 

With both of these relations, we have established that as $\Phi_0$ decreases, $\min{(t^*,\tau)}$ increases. To complete the proof, observe that
$t^* \leq \tau$ implies that $t_f = t^*$, and thus as $\Phi_0$ decreases, $\min{(t_f,\tau)}$ increases.
\end{proof}

Having handled $\Phi_0 < \frac{1}{m}$ up to the transition surface, we now turn to $\Phi \geq \frac{1}{m}$. 
These solutions are either a constant strategy for the whole game, when $\Phi_0 > \frac{1}{m}$, or they are a 
continuation of a two-phase strategy from the transition surface to the terminal surface. 
The initial conditions for this case are given
by a vertical line above the transition surface and continued to the right along the transition surface, as presented in
Figure~\ref{fig:TrajectoriesB}. We will start with those initial conditions on the vertical. 

\begin{lemma}
\label{lemma:SI2}
Let $I_0 \geq \Phi_0 \geq \frac{1}{m}$. Then 
\begin{equation}
t_f = \frac{m(I_0 - \Phi_0)}{I_0},
\end{equation}
and thus $t_f$ is monotonically decreasing in $\Phi_0$. 
\end{lemma}

\begin{proof}
$\Phi_0 \geq \frac{1}{m}$ implies that $c^* = \frac{1}{m}$ for the duration of the game. 
In this case, $I(t)$ is constant and 
from the initial condition $(I_0, \Phi_0)$, using Eq.~\eqref{eq:Phisol} we have
\begin{equation}
\label{eq:PhiLint}
\Phi(t) = \Phi_0 + \frac{I_0}{m} t. 
\end{equation}

At the terminal time $t_f$, $\Phi(t_f) = I(t_f) = I_0$, and inserting these into Eq.~\eqref{eq:PhiLint} gives
\begin{equation}
t_f = \frac{m(I_0 - \Phi_0)}{I_0},
\end{equation}
which is clearly decreasing in $\Phi_0$. 
\end{proof}

As stated, this lemma demonstrates that $t_f(I_0,\Phi_0)$ is monotonically decreasing in $\Phi_0$.
This makes intuitive sense as the further $\Phi_0$ is away from the terminal condition, the longer it will take
to get there because trajectories overlap in this case.

The final relation to establish is between  $t_f$ and $\tau$ when
 $\Phi = \frac{1}{m}$ and $\tau \geq 0$, corresponding to $I(\tau) \in [I_0,1]$.
\begin{lemma}
\label{lemma:SI3}
Fix $I_0>0$ and suppose $0<\Phi_0 \leq \frac{1}{m}$ such that $0 \leq \tau(I_0,\Phi_0) < \infty$. 
Then as $\tau(I_0,\Phi_0)$
increases, $t_f(I_0, \tau(I_0,\Phi_0))$ is strictly monotone increasing. 
\end{lemma}

\begin{proof}
To begin, fix $I_0$ and let $0 < \Phi_0 < \frac{1}{m}$ such that  $\tau > 0$. $I(t)$ is constant for all $t \geq \tau$, so using 
Eq.~\eqref{eq:Ifunc} we have
\begin{equation}
\label{eq:SI3intermidate1}
I(t) = I(\tau) = \frac{ \frac{I_0}{1-I_0} \exp{( \tau)}}{1 + \frac{I_0}{1-I_0} \exp{( \tau)}}.
\end{equation}

From Eq.~\eqref{eq:Phisol} for $t \geq \tau$, we have
\begin{equation}
\Phi(t) = \frac{1}{m} + \frac{I(\tau)}{m} (t - \tau)
\end{equation}

Using the terminal condition $\Phi(t_f) = I(\tau)$ and rearranging, 
\[
1 = \frac{1}{m I(\tau)} + \frac{1}{m}(t_f - \tau). 
\]
Using Eq.~\eqref{eq:SI3intermidate1}, this becomes
\begin{equation}
1 = \frac{(1 - I_0) \exp{(-\tau)} + I_0}{m I_0} + \frac{1}{m} (t_f - \tau). 
\end{equation}

Taking the derivative in $\tau$ we obtain
\begin{equation}
\frac{d t_f}{d \tau} = 1 + \frac{(1- I_0) \exp{(-\tau)}}{I_0} > 0
\end{equation}
for $\tau > 0$. This establishes that $t_f$ is strictly monotone increasing in $\tau$. 
\end{proof}

We now have everything we need to prove the main theorem of this section.

\begin{theorem}
	There is always exactly one solution to the
	necessary conditions \eqref{eq:IPhiopt}'s 
	for a subgame perfect
	Nash equilibrium of 
	the \SISDG.
\end{theorem}

\begin{proof}
To establish uniqueness, we need to demonstrate that for each $t_f \in [0,\infty)$, there exists a single $\Phi_0 \in (0, I_0]$ 
that satisfies $\Phi(t_f)=I(t_f)$. We see from inspection that $\dot{\Phi} > 0$. This implies
that solutions are either constant strategy single-phase, or two-phase with a single change from $0$ to $\frac{1}{m}$. 

Using the lemmas proven above, we must demonstrate that as $\Phi_0$ decreases, 
$t_f(I_0,\Phi_0)$ increases strictly monotonically. First, taking $\Phi_0 = I_0$ corresponding to the trivial solution 
$t_f = 0$, there always exists a $t_f<\infty$ with a solution. 

Following our previous discussion, there are two situations for this system based around how $I_0$ relates to $\frac{1}{m}$. 
First, when $I_0 < \frac{1}{m}$, all admissible $\Phi_0$ satisfy $\Phi_0 < \frac{1}{m}$, so $c^*(0) = 0$. 

Using Lemma~\ref{lemma:SI1}, we established $\min{(t_f(\Phi_0),\tau(\Phi_0))}$ is strictly monotone decreasing in $\Phi_0$. As described, this gives monotonicity between the $\Phi_0$ on the initial line $I_0$ to the line
$\min{\big(I_0, \frac{1}{m}\big)}$. 

To partition those $\Phi_0$ that have single-phase equilibria and those that have two-phase, 
let $\Phi_0^*$ be such that $\Phi_0^*(t_f) = \frac{1}{m}$, so $\tau = t_f$. 
Since the system is autonomous, trajectories do not cross, so this gives a partition. 
For $\Phi_0 > \Phi_0^*$, $t_f < \tau$. 
In this case, we have explicitly $\frac{d t_f}{d \Phi_0}$ < 0 from Lemma~\ref{lemma:SI1}, giving the strict monotone 
relationship we desire.
For $\Phi_0 \leq \Phi_0^*$, $\tau \leq t_f$. In this case
we have $\frac{d \tau}{d \Phi_0} < 0$, also from Lemma~\ref{lemma:SI1}. 

$\Phi_0^* \in (0, I_0]$ must exist since from Eq.~\eqref{eq:Phisol} when $c^*=0$, there is a $\Phi_0$ for each $\tau \in 
[0,\infty)$. $\Phi_0^*$ corresponds to the $\tau$ where $I(\tau) = \frac{1}{m}$, which exists due to Eq.~\eqref{eq:Ifunc}. 

To complete the proof of this case, Lemma~\ref{lemma:SI3} establishes $\frac{d t_f}{d\tau}>0$. Thus, when $\tau \geq 
0$, as $\Phi_0$ decreases, $\tau$ increases strictly monotonically, which implies that $t_f$ increases strictly 
monotonically.
Since the relation between $t_f$ and $\Phi_0$ is strictly monotone, there is a unique initial $\Phi_0$ for each terminal time $t_f$. 

The second situation, when $I_0 \geq \frac{1}{m}$, has single-phase strategies where $c^* = \frac{1}{m}$
for $\Phi_0 \geq \frac{1}{m}$.

For all $\Phi_0 \in (0,\frac{1}{m}]$, $\tau \geq 0$ and the previous argument still holds. 
Otherwise, Lemma~\ref{lemma:SI2} establishes $\frac{d t_f}{d \Phi_0} < 0$.
Since the two approaches to $t_f$ coincide at $\Phi_0 = \frac{1}{m}$, we have again established a continuous 
bijection between $\Phi_0$ and $t_f$. 

Lastly, to establish that a $\Phi_0$ exists for each $t_f \in [0, \infty)$, suppose there exists an $M$ such that 
$t_f(\Phi_0) < M$ for all $\Phi_0 \in (0,I_0]$. 
Observe that
taking $\Phi_0 = 0$ gives a trajectory which tends to $(1,0)$ as $t \to \infty$ , so it never reaches terminal condition in 
finite-time. 

For any $m>0$,  while $\Phi < \min{\big( I_0, \frac{1}{m}\big)}$ the equilibrium is constant $c^* = 0$.
Eq.~\eqref{eq:Phisol} gives the formula for $\Phi(t)$ when $c^*=0$ of
\[
\Phi(t) = e^t \Phi_0. 
\]
To show that there exists a $\Phi_0$ corresponding to a game of duration at least $M+1$ for all $M>0$, 
we can write
\[
e^{-(M+1)} \min{\Big( I_0, \frac{1}{m}\Big)} = \Phi_0. 
\]
This $\Phi_0$ gives a solution trajectory which either terminates or hits the transition surface at time $M+1$. Thus
there cannot be an upper bound, so for all $t_f \in [0,\infty)$ there exists a corresponding $\Phi_0$.
\end{proof}

While our proof does not make explicit use of the relation between $t_f$ and $\Phi_0$, instead focusing on establishing how 
one varies as the other changes which is sufficient to establish a bijective relationship, 
we can construct the relation for the three 
types of equilibria in this case. In the proofs of Lemmas~\ref{lemma:SI1},~\ref{lemma:SI2},~\ref{lemma:SI3}
we construct all the equations needed to piece together the following relations.

\begin{corollary}
	Let $\Phi_0 \in (0, I_0]$. If $c^*$ satisfying the \SISDG\xspace is single-phase, then
	\begin{align*}
	t_f(\Phi_0) &= \ln{\bigg( \frac{1}{\Phi_0 I_0} - \frac{1-I_0}{I_0} \bigg)} \text{  when } c^* = 0, \\
	t_f(\Phi_0) &= \frac{m(I_0 - \Phi_0)}{I_0} \text{  when } c^* = \frac{1}{m}.
	\end{align*}
	When $c^*$ is two-phase,
	\begin{equation*}
		t_f(\Phi_0) = m - \ln{(m \Phi_0)} - \frac{m(1-I_0)\Phi_0 + I_0}{I_0}.
	\end{equation*}
\end{corollary}
One can verify that these are monotone, and that they agree for those $\Phi_0$ which are the boundary between single-phase 
and two-phase strategies.

Having presented our argument for uniqueness of the Nash equilibria in the \SISDG\xspace in full, it is worth offering a brief
discussion of how our approach depends on the assumptions we have made in setting up the model. A key property that
we leverage is monotonicity of $\Phi$, which makes the Nash equilibrium $c^*$ monotonically increasing, in this case
changing strategy at most once as the trajectory crosses between the two regions of constant $c^*$. 
Monotonicity of $\Phi$, as can be discerned from
the definition in Eq~\ref{eq:phidef}, is due entirely to monotonicity of $I$ when there is no future-cost discounting ($h = 0$). 
While this property is not universal in the broader
social distancing game space, for example the SIR variant does not have monotone $c^*$ \citep{bib:Reluga2013}, it is a general
property in the SI variant. Cost discounting can complicate this picture because $V$ is no longer necessarily increasing, 
but as we demonstrate in our forthcoming work on this model with $h>0$, it ultimately does not alter the monotonicity of $\Phi$ in this variant. 

The fundamental challenge of our approach to proving uniqueness is extracting sufficient information about the duration of individual
trajectories and its relationship to the cost of the game so that we can establish monotonicity. 
The regions of constant $c^*$ allow us to cut up a trajectory into segments through each region, over which we can then
integrate $I$ and $\Phi$ sufficiently to piece together the full trajectory and study its dependence on the initial conditions. 
Without regions of constant $c^*$ which are granted by the threshold-linear form of $\sigma$,
it is unclear how to examine the trajectory duration because we cannot integrate $I$ and $\Phi$ while $c^*$ is changing continuously. 
So while this part of our argument
should be generalizable to piecewise-linear $\sigma$ since these only introduce more regions of constancy that can similarly 
be glued together, it does not clearly apply to the continuously changing equilibria that arise from a strictly convex $\sigma$.
 Specific examples of strictly convex $\sigma$ might be found
where the integrals of $I$ and $\Phi$ can be computed sufficiently and thus uniqueness established through our approach, 
but the general question of uniqueness in the \SISDG\xspace can only seemingly be settled using our approach by way of
a limiting argument from piecewise-linear $\sigma$ games.

\section{Discussion}

Here we have presented the first formal proof of the uniqueness of Nash equilibria in the finite time social distancing game as 
studied by \citet{bib:Reluga2010,bib:Reluga2013} for an SI epidemic. These equilibria strategies are two-phase off-on strategies,
which reflect the linearity of $(1-mc)^+$ and the monotonicity of the risk of infection. 
The method of analysis we have deployed to establish uniqueness in this simplified case will be applied analyze to the positive 
discounting case $(h>0)$ in a forthcoming work, and our ultimate goal is to characterize the Nash equilibria structure for the 
SI social distancing game for any admissible convex $\sigma$. As we have explained above, the clear path forward is to
generalize our argument to piecewise-linear $\sigma$, which, like the threshold-linear form we have focused on, divide the 
$I,\Phi$ phase-plane into regions of constant $c^*$. While this will require wrestling with additional details, the fundamental
aspects of the proof of uniqueness are all developed above. With this case handled, it should be possible to construct a limiting
argument for games with any admissible $\sigma$ using these piecewise-linear games, and will settle the question of 
uniqueness in the SI social distancing game. 

Our proof focused on the perfect intervention case because of its mathematical simplicity, but one cannot expect that social distancing interventions are perfect in many cases. While slightly more difficult, we expect that within this approach one can prove uniqueness for imperfect interventions, where $\sigma(c) =  \max{(1 - mc, b)}$ for $b \in [0,1)$. 
In lieu of a proof, we observe that the only difference is that 
$\dot{I} > 0$ when $\Phi > \frac{1}{m}$, which should not affect the conclusions but may make the equations more difficult
to analyze. Our as-yet-unpublished results with $h >0$ do suggest that this conjecture could be provable within this framework. 

The second aspect of our analysis was investigating the efficacy of the Nash equilibria strategies, which we did by restricting
the strategy space to only on-off strategies. A key insight from this analysis is social distancing in this game is only
efficacious for epidemics which go on for some duration but a vaccine arrives before too long, dependent on the efficiency of 
the intervention. We also demonstrated that the \DSISDG\xspace does not have a free-riding effect, which corroborates a conclusion 
from \citet{bib:Chen2012}. 

Because of our ability to explicitly construct the \DSISDG\xspace equations, we were able to prove that the Nash equilibrium in the
delay strategy space is a global ESS (Appendix~\ref{secA2}, Theorem~\ref{thm:ESS}). While we would like to extend this 
characterization to the general strategy space, even a two-dimensional strategy space can make this analysis very 
difficult, as seen by \citet{bib:Li2013}. We are unaware of any successful efforts to codify the concept of an ESS for an 
infinite-dimensional strategy space or to develop the tools to characterize an ESS over such a space. This system may provide
the perfect mixture of complexity and simplicity to undertake such an effort. 

Our results here are strongly dependent on the monotonicity of risk in the SI
epidemic dynamic.  They will certainly not generalize in full to SIR theory,
where generally risk both increases and decreases throughout the epidemic.  
Any attempt to analyze even the SIR variant, like
in \citep{bib:Reluga2013}, quickly encounters difficulties that we avoided. 
An interesting generalization worth considering is replacing $I(t)$ with any 
monotonically increasing function $0 < J(t) \leq 1$. While the general structure of our
analysis should be maintained under this modification, and thus one expects to be able to 
establish the requisite monotonicity if the lack of explicit equations can be overcome, 
evidence suggests that the solution structure of the \DSISDG\xspace is not 
stable under this modification (Appendix~\ref{secA3}). While the \DSISDG\xspace and
the \SISDG\xspace have equivalent equilibria in the SI setting,
it may be that they cease to coincide in the general setting, 
or that the solution concept for the \SISDG\xspace similarly breaks down upon generalization, 
but more work needs to be done to bring this fully into frame.

Finally, while our analysis of the \SISDG\xspace is an important piece of
progress in our understanding of epidemic games, there is still much work to
do.  As mentioned in the introduction, the \SISDG\xspace as presented here
shares elements with many practical epidemic management problems but is not an
exact match to any we know.  Any attempts at application should carefully
consider departures from our assumptions and the particular epidemiological
scenario of interest.  HIV, for example, is a reasonable match to the SI
dynamic, as infected individuals are never again virus-free, but there is no
vaccine or cure that could create a finite bound on the epidemic, people can
not be assumed to know their own infection status, the time-scale of HIV spread
is so long that demographic turnover can not properly be neglected, and
widespread access to PReP has changed the incentive structure
\citep{bib:Mody2024}.  Hepatitis B can become a chronic infection that is is
incurable but is preventable with vaccination.  But Hepatitis B is typically
cryptic, becomes chronic in only 10 percent of cases, and is only highly
transmissible at birth (not SI-like) or in certain drug-using communities
(which by nature have perverse incentive structures)
\citep{bib:MacLachlan2015}.  There are disease like measles, mumps, and
smallpox, for which people do know their own infection status and mass
vaccination is a traditional outbreak-response, but these are all SIR-like
infections, where the window of infectivity is often much shorter than the
outbreak duration \citep{bib:Hall2021}.  The recovery of infected individuals
leads to an eventual decline in infection risk, which breaks the monotonicity
assumption fundamental to our current analysis.  Computational results do
exist for SIR social-distancing games \citep{bib:Reluga2010},
but the establishment of a general synthetic proof of uniqueness remains an
open problem.

\section*{Acknowledgments}

We thank two anonymous reviewers from their constructive comments, particularly
with respect to our free-riding free result and clarification of the equilibrium
structure conditions.

\section*{Declarations}

\bmhead{Funding} C. Olson was supported by the National Science Foundation Graduate Research Fellowship Program under Grant No. DGE1255832 while producing this work. Any opinions, findings, and conclusions or recommendations expressed in this material are those of the author and do not necessarily reflect the views of the National Science Foundation.

\bmhead{Conflicts of Interest} The authors have no relevant conflicts of interest to declare. 

\bmhead{Data Availability} Data generated during the production of this work is available upon reasonable request to the contributing author.

\begin{appendices}

\section{Minima of Special Functions}
\label{secA1}

For our proof of the ESS conditions on the delay strategy space, we
encounter two functions for which the Nash strategy is a global minimum to
each, and this allows us to establish the necessary inequalities. Here, we
present proofs of both facts. 

\begin{lemma}
	\label{lemma:firstmin}
	Consider $F: [0,t_f] \rightarrow \RR$ such that
	\begin{equation}
		F(x) = \exp(cx)\Big(\frac{x}{m} -1 \Big)
	\end{equation}
	where $c > 0$ and $m>0$.
	If $x^* := m - 1/c \in [0,t_f]$, then $x^*$ is the global minimizer of $F$.
\end{lemma}

\begin{proof}
	Observe that if $L(y) = y e^y$, then $F(x) = e^{cm}/(cm) L((x-m)c)$.
	Now $L(y)$ is a partial inverse of the Lambert-Corless $\operatorname{W}$-function \citep{bib:Corless96},
	and is well known to have a global minimum at $y^*:=-1$.  Thus,
	$F(x)$ must have a global minimum at $x^*=m-1/c$, provided $x^* \in [0,t_f]$.
\end{proof}

\begin{lemma}
	\label{lemma:secondmin}
	Let $I(t)$ be as given in Eq.~\eqref{eq:Ifunc}. Suppose there exists an $x \in (0, t_f)$ with $I(t_f - x) = \frac{1}{m-x}$, then $x$ is a global minimum of 
	\begin{equation}
		(1 - I(t_f - x))\Big(\frac{x}{m}-1 \Big).
	\end{equation}
\end{lemma}

\begin{proof}
	The critical point condition is 
	\begin{equation}
		0=(1- I(t_f-x))\Big(\frac{1}{m} + I(t_f-x)\Big(\frac{x}{m} - 1\Big)\Big),
	\end{equation}
	which is satisfied when $x$ satisfies $I(t_f - x) = \frac{1}{m-x}$. 
	The second derivative evaluated at the $x$ satisfying the condition is 
	\begin{equation}
		\frac{(1 - I(t_f - x)I(t_f-x))}{1-I_0}\Big(1 - \frac{x}{m}\Big)>0.
	\end{equation}
	because $\frac{1}{m-x} = I < 1$, which implies that $\frac{x}{m} < 1$. 
\end{proof}

\section{Differentiability of the Delay Disutility}\label{secA3}

In order to apply first-order derivative conditions in establishing our equilibrium results
in the \DSISDG, we must know that the delayed disutility is continuous and differentiable.
Special care must be taken because of the disutility is expressed
in terms of Max functions, which have cusps generally.
\begin{theorem} \label{thm:DCD}
	For the \DSISDG, $\D(x,\ox)$ is both continuous and differentiable over $[0,t_f]\times [0,t_f]$
	for all initial conditions $I_0 \in [0,1]$.
\end{theorem}

Before proceeding with the proof,
we caution the reader that Theorem~\ref{thm:DCD} {\bf is not generic}.  While $\D(x,\ox)$
is generally continuous, its globally differentiability is a structurally unstable
accident of the mass-action SI hypothesis and the stiffness of the delay strategy space!  While we conjecture
that our macroscopic results hold very generally, small
structural perturbations can break the differentiability of $\D(x,\ox)$ and lead
to very different equilibrium microstructure, including loss of classical existence and
uniqueness!  Full consideration is beyond our current scope and must be left for future research.
\citet{bib:Bouveret2020} provide a starting point for further consideration in terms
of relaxed equilibria concepts.

\begin{proof}
In our construction of the delay-strategy game,
the strategies $\zeta(t;x)$ have discontinuous jumps in value when $t = t_f - x$.
Since $\D(x,\ox) = D(\zeta(t;x),\zeta(t,\ox))$, we need to consider variation of both $x$ and $\ox$.  
\DSISDG\xspace says $\D(x,\ox) = 1 - q(x,\ox) (1-x/m)$,
with a special limiting case when $I_0 = 1$.
Since both $\max(x,\ox)$ and $(x-\ox)^+$ are continuous functions, 
the continuity of $\D$ in $x$ and $\ox$ follows the definition of $q(x,\ox)$ given
in Eq.~\eqref{eq:DSISDGc} and the composition of continuous functions.

To demonstrate differentiability of Eq.~\eqref{eq:DSISDGb}, we have the cases of
$I_0=1$ and $I_0 < 1$. The case of $I_0=1$ is trivial.  For the other case of $I_0<1$,
it suffices (by inspection) to show that $q(x,\ox)$ is differentiable, in the sense that
there is a unique tangent plane at each strategy pair $(x,\ox)$.  
In neighborhoods of $x=\ox$, the  mild irregularity of the Max function
appears, and we must consider the specifics of the surface. When $x < \ox$, 
\begin{subequations}
	\label{eq:qgradA}
	\begin{align}
		\frac{d q}{d x} &= \frac{I(t_f - \ox)\Big( 1 - I(t_f - \ox) \Big)}{1-I_0} \exp{\Big(-I(t_f - \ox)(\ox-x) \Big)},
		\\
		\frac{d q}{d \ox} &= \frac{I(t_f - \ox)\Big(1 - I(t_f - \ox) \Big) (\ox - x)}{1 - I_0} \exp{\Big(-I(t_f - \ox) (\ox-x) \Big)}.
	\end{align}
\end{subequations} 
When $x > \ox$,
\begin{subequations}
	\label{eq:qgradB}
	\begin{align}
		\frac{d q}{d x} &= \frac{I(t_f - x)\Big(1 - I(t_f - x) \Big)}{1-I_0},
		\\
		\frac{d q}{d \ox} &= 0.
	\end{align}
\end{subequations}
Taking the one-sided limits as $x \to \ox$ from below, we observe that the
derivatives \eqref{eq:qgradA} converges to the derivatives
\eqref{eq:qgradB}. Since the derivatives of these planes exist in
neighborhoods of $x=\ox$ and are equal in the limit of $x = \ox$, then we
must have a uniquely defined tangent plane along $x=\ox$. Away from the
line $x=\ox$, standard results imply existence and uniqueness. Thus,
$\D(x,\ox)$ is everywhere differentiable.
\end{proof}



An immediate consequence of these calculations is the following.
\begin{theorem}
	The \DSISDG\xspace is free from free-riding.
\end{theorem}
\begin{proof}
	The rate of change of the emblematic disutility is
	\begin{equation}
		\frac{d \mathcal{E}}{d\ox} =
		\left. \left( \frac{d\D}{dx} + \frac{d\D}{d\ox} \right) \right|_{x = \ox }.
	\end{equation}
	Since, by our preceding argument, $x=\ox$ implies ${d \D}/{d \ox} = {d q}/{d \ox} = 0$,
	and ${d \mathcal{E}}/{d\ox} = {d \D}/{dx}$.
	Thus, strategies satisfying the necessary conditions for the social optima
	also satisfy the necessary conditions for a Nash equilibrium.  The equivalence
	then follows from the geometry of Eq.~\eqref{eq:emblem} and the uniqueness of
	the Nash equilibrium.  Since the socially optimal strategy is also
	the unique Nash equilibrium, there can be no free-riding.
\end{proof}

\section{ESS Proof}\label{secA2}

In this appendix we prove that the Nash equilibrium for the \DSISDG\xspace is a global Nash equilibrium and in fact an ESS \citep{bib:MaynardSmith1974}. To do this, we lean upon a relation established through the local necessary
condition for optimality.  If $\D(x,\ox)$ is continuous with
a continuous derivative, and $x^*$ is a best-response to itself, then
\begin{equation}
	0 = \left. \frac{d\D}{dx} \right|_{x = x^*, \ox=x^* }. 
\end{equation}
This is equivalent to $x^*$ satisfying
\begin{equation}
	0 = \frac{1 - I(t_f -x^*)}{m(1-I_0)} \Big((m-x^*)I(t_f - x) - 1 \Big).
\end{equation}
After simplifying, we find
\begin{equation}
	I(t_f - x^*) = \frac{1}{m-x^*}. \label{eq:Imx}
\end{equation}

To establish that the local Nash equilibrium given by $x^*$ satisfying the local condition Eq.~\eqref{eq:Imx} is an ESS, 
we prove that it is a global Nash equilibrium with invasion potential. We establish these through the two inequalities given 
in the following theorem. 

\begin{theorem}
	\label{thm:ESS}
	Let $x^* \in (0,t_f)$ satisfy Eq.~\eqref{eq:Imx}. Then for all $y \in [0,t_f]$, $y \not = x^*$, 
	\begin{subequations}
		\begin{align} 
			\D(x^*,x^*) &< \D(y,x^*), \label{eq:NashCondition} \\
			\D(x^*, y) &< \D(y,y). \label{eq:InvCond}
		\end{align}
	\end{subequations}
\end{theorem}

Eq.~\eqref{eq:NashCondition} is the Nash equilibrium condition and Eq.~\eqref{eq:InvCond} is the invasion condition. Together these 
establish $x^*$ as an ESS. 

\begin{proof}
	To establish Eq.~\eqref{eq:NashCondition}, we analyze
	\begin{equation}
		\D(y,x^*) - \D(x^*,x^*) = q(y,x^*)\Big(\frac{y}{m} - 1\Big) - q(x^*,x^*)\Big(\frac{x^*}{m} - 1\Big).
	\end{equation}
	To establish the desired inequality, we first presume $y < x^*$. In this case, 
	\begin{gather}
		\D(y,x^*) - \D(x^*,x^*)
		= \frac{1-I(t_f-x^*)}{1-I_0}
		\Big( \exp{(I(t_f-x^*)(y-x^*))}\Big(\frac{y}{m}-1\Big) - \frac{x^*}{m}+1 \Big).
	\end{gather}
	Taking $C= I(t_f - x^*)$ and applying Lemma~\ref{lemma:firstmin}, we see that $x^*$ is the global minimizer, so the previous difference is positive. Therefore,
	\begin{equation}
			\D(y,x^*) - \D(x^*,x^*) >0. 
	\end{equation}
	
	In the second case of $x^* < y$, we have
	\begin{equation}
			\D(y,x^*) - \D(x^*,x^*) = \frac{1-I(t_f-y)}{1-I_0} \Big(\frac{y}{m} - 1\Big) - \frac{1-I(t_f-x^*)}{1-I_0} \Big(\frac{x^*}{m} - 1\Big).
	\end{equation}
	Applying Lemma~\ref{lemma:secondmin}, $x^*$ is the global minimum, so we have
	\begin{equation}
			\D(y,x^*) - \D(x^*,x^*) >0,
	\end{equation} 
	which establishes the Nash equilibrium condition. 
	
	To prove the invasion condition Eq.~\eqref{eq:InvCond}, we similarly analyze
	\begin{equation}
		\D(y,y) - \D(x^*,y) = q(y,y)\Big(\frac{y}{m}-1\Big) -  q(x^*,y)\Big(\frac{x^*}{m}-1\Big).
	\end{equation}
	Again, we separate this into cases. If $x^*<y$, 
	\begin{equation}
		\D(y,y) - \D(x^*,y) = \frac{1 - I(t_f - y)}{1-I_0} \Big(\frac{y}{m}-1 - \exp{(I(t_f-y)(x^*-y))}\Big(\frac{x^*}{m}-1\Big)\Big).
	\end{equation}
	Applying Lemma~\ref{lemma:firstmin} with $C = (I(t_f-y)$ demonstrates that $x^*$ is the global minimum, and therefore
	\begin{equation}
			\D(y,y) - \D(x^*,y)>0.
	\end{equation}
	Lastly, when $x^*>y$, 
	\begin{equation}
			\D(y,y) - \D(x^*,y) = \frac{1 - I(t_f-y)}{1-I_0} \Big(\frac{y}{m}-1\Big) - \frac{1 - I(t_f-x^*)}{1-I_0} \Big(\frac{x^*}{m}-1\Big).
	\end{equation} 
	Applying Lemma~\ref{lemma:secondmin}, we obtain
	\begin{equation}
		\D(y,y) - \D(x^*,y)>0,
	\end{equation}
	which establishes the invasion condition. 
\end{proof}

\end{appendices}

\bibliography{Biblio-Database}

\end{document}